\documentclass[11pt, eqref, reqno]{amsart}
\usepackage{url}
\usepackage{color}
\usepackage[colorlinks,linkcolor=blue,citecolor=red]{hyperref}
\allowdisplaybreaks[4]
\newcommand{\cprime}{\/{\mathsurround=0pt$'$}}

\newcommand*{\taup}[1]{{#1}^{\mathbf{P}}}
\newcommand*{\taui}[1]{{#1}^{\mathbf{I}}}

\let \phi = \varphi

\DeclareMathOperator{\const}{const}

\usepackage[all]{xy}
\SelectTips{cm}{}
\usepackage{mathrsfs}
\let\mathcal\mathscr
\usepackage[mathscr]{eucal}

\theoremstyle{theorem}
\newtheorem{theorem}{Theorem}
\newtheorem{proposition}{Proposition}

\theoremstyle{definition}
\newtheorem*{coord}{Coordinates}
\newtheorem{example}{Example}
\theoremstyle{remark}
\newtheorem{remark}{Remark}

\DeclareMathOperator{\rank}{rank}
\newcommand*{\sd}[2]{\{\,#1\mid#2\,\}}
\newcommand*{\eval}[1]{\left.#1\right|}
\newcommand*{\abs}[1]{\left|#1\right|}
\newcommand*{\Ev}{\mathbf{E}}

\date{\today}

\title[Nonlocal conservation laws]{Nonlocal conservation laws of PDEs possessing
  differential coverings}

\author{I.~Krasil{\cprime}shchik}\address{V.A.~Trapeznikov Institute of
  Control Sciences RAS, Profsoyuznaya 65, 117342 Moscow,
  Russia}\email{josephkra@gmail.com}\thanks{The work was partially
  supported by the RFBR Grant 18-29-10013 and IUM-Simons Foundation.} 
\begin{document}

\begin{abstract}
  In his 1892 paper~\cite{Bianchi}, L.~Bianchi noticed, among other
  things, that quite simple transformations of the formulas that describe the
  B\"{a}cklund transformation of the sine-Gordon equation lead to what is
  called a nonlocal conservation law in modern language. Using the techniques
  of differential coverings~\cite{Trends}, we show that this observation is
  of a quite general nature. We describe the procedures to construct such
  conservation laws and present a number of illustrative examples.
\end{abstract}



\dedicatory{To the memory of Alexandre Vinogradov, my teacher}
\maketitle
\tableofcontents

\section*{Introduction}

In~\cite{Bianchi}, L.~Bianchi, dealing with the celebrated B\"{a}cklund
auto-transformation\footnote{I changed the original notation slightly}
\begin{equation}\label{eq:1}
  \dfrac{\partial(u - w)}{\partial x} = \sin(u +
  w),\quad
  \dfrac{\partial(u + w)}{\partial y} = \sin(u - w)
\end{equation}
for the sine-Gordon equation
\begin{equation}\label{eq:6}
  \frac{\partial^2(2u)}{\partial x\partial y} = \sin (2u)
\end{equation}
in the course of intermediate computations (see~\cite[p.~10]{Bianchi}) notices
that the function
\begin{equation*}
  \psi = \ln \frac{\partial u}{\partial C},
\end{equation*}
where $C$ is an arbitrary constant on which the solution~$u$ may depend,
enjoys the relations
\begin{equation*}
  \frac{\partial\psi}{\partial x} = \cos(u + w),\quad
  \frac{\partial\psi}{\partial y} = \cos(u - w).
\end{equation*}
Reformulated in modern language, this means that the $1$-form
\begin{equation*}
  \omega = \cos(u + v)\,dx + \cos(u - v)\,dy
\end{equation*}
is a nonlocal conservation law for Eq.~\eqref{eq:1}.

Actually, Bianchi's observation is of a very general nature and this is shown
below.

In Sec.~\ref{sec:preliminaries}, I shortly introduce the basic constructions
in nonlocal geometry of PDEs, i.e., the theory of differential
coverings,~\cite{Trends}. Sec.~\ref{sec:main-result} contains an
interpretation of the result by L.~Bianchi in the most general setting. In
Sec.~\ref{sec:examples}, a number of examples is discussed.

Everywhere below we use the notation~$\mathcal{F}(\cdot)$ for the
$\mathbb{R}$-algebra of smooth functions, $D(\cdot)$ for the Lie algebra of
vector fields, and $\Lambda^*(\cdot) = \oplus_{k\geq0}\Lambda^k(\cdot)$ for
the exterior algebra of differential forms.

\section{Preliminaries}\label{sec:preliminaries}

Following~\cite{AMS-book}, we deal with infinite
prolongations~$\mathcal{E} \subset J^\infty(\pi)$ of smooth submanifolds
in~$J^k(\pi)$, where~$\pi\colon E\to M$ is a smooth locally trivial vector
bundle over a smooth manifold~$M$, $\dim M = n$, $\rank\pi = m$.
These~$\mathcal{E}$ are differential equations for us. Solutions
of~$\mathcal{E}$ are graphs of infinite jets that lie in~$\mathcal{E}$. In
particular, $\mathcal{E} = J^\infty(\pi)$ is the tautological equation~$0=0$.

The bundle $\pi_\infty\colon \mathcal{E} \to M$ is endowed with a natural flat
connection $\mathcal{C}\colon D(M) \to D(\mathcal{E})$ called the Cartan
connection. Flatness of~$\mathcal{C}$ means that~$\mathcal{C}_{[X,Y]} =
[\mathcal{C}_X, \mathcal{C}_Y]$ for all $X$, $Y\in D(M)$. The distribution
on~$\mathcal{E}$ spanned by the fields of the form~$\mathcal{C}_X$ (the Cartan
distribution) is Frobenius integrable. We denote it by~$\mathcal{C}\subset
D(\mathcal{E})$ as well.

A (higher infinitesimal) symmetry of~$\mathcal{E}$ is a $\pi_\infty$-vertical
vector field $S\in D(\mathcal{E})$ such that $[X,\mathcal{C}] \subset
\mathcal{C}$.

Consider the submodule $\Lambda_h^k(\mathcal{E})$ generated by the
forms~$\pi_\infty^*(\theta)$, $\theta\in \Lambda^k(M)$. Elements $\omega \in
\Lambda_h^k(\mathcal{E})$ are called horizontal $k$-forms. Generalizing
slightly the action of the Cartan connection, one can apply it to the de~Rham
differential $d\colon \Lambda^k(M) \to \Lambda^{k+1}(M)$ and obtain the
horizontal de~Rham complex
\begin{equation*}
  \xymatrix{
    0\ar[r]& \mathcal{F}(\mathcal{E})\ar[r]& \dots\ar[r]&
    \Lambda_h^k(\mathcal{E})\ar[r]^-{d_h}& \Lambda_h^{k+1}(\mathcal{E})\ar[r]&
    \dots\ar[r]& \Lambda_h^n(\mathcal{E})\ar[r]&0
  }
\end{equation*}
on~$\mathcal{E}$. Elements of its $(n-1)$st cohomology
group~$H_h^{n-1}(\mathcal{E})$ are called conservation laws
of~$\mathcal{E}$. We always assume~$\mathcal{E}$ to be differentially
connected which means that~$H_h^0(\mathcal{E}) = \mathbb{R}$.

\begin{coord}
  Consider a trivialization of~$\pi$ with local coordinates $x^1,\dots,x^n$
  in~$\mathcal{U}\subset M$ and $u^1,\dots,u^m$ in the fibers
  of~$\eval{\pi}_{\mathcal{U}}$. Then in $\pi_\infty^{-1}(\mathcal{U}) \subset
  J^\infty(\pi)$ the adapted coordinates~$u_\sigma^i$ arise and the Cartan
  connection is determined by the total derivatives
  \begin{equation*}
    \mathcal{C}\colon \frac{\partial}{\partial x^i} \mapsto D_i =
    \frac{\partial}{\partial x^i} + \sum_{j,\sigma}u_{\sigma
      i}^j\frac{\partial}{\partial u_\sigma^j}.
  \end{equation*}
  Let $F = (F^1,\dots,F^r)$, where~$F^j$ are smooth functions
  on~$J^k(\pi)$. The the infinite prolongation of the locus
  \begin{equation*}
    \sd{z\in J^k(\pi)}{F^1(z) = \dots = F^r(z) = 0} \subset J^k(\pi)
  \end{equation*}
  is defined by the system
  \begin{equation*}
    \mathcal{E} = \mathcal{E}_F = \sd{z\in J^\infty(\pi)}{D_\sigma(F^j)(z) =
      0,\ j=1,\dots,r,\ \abs{\sigma}\geq0},
  \end{equation*}
  where~$D_\sigma$ denotes the composition of the total derivatives
  corresponding to the multi-index~$\sigma$. The total derivatives, as well as
  all differential operators in total derivatives, can be restricted to
  infinite prolongations and we preserve the same notation for these
  restrictions. Given an~$\mathcal{E}$, we always choose internal local
  coordinates in it for subsequent computations. To restrict an operator
  to~$\mathcal{E}$ is to express this operator in terms of internal
  coordinates.

  Any symmetry of~$\mathcal{E}$ is an evolutionary vector field
  \begin{equation*}
    \Ev_\phi = \sum D_\sigma(\phi^j)\frac{\partial}{\partial u_\sigma^j}
  \end{equation*}
  (summation on internal coordinates), where the functions
  $\phi^1,\dots,\phi^m \in \mathcal{F}(\mathcal{E})$ satisfy the system
  \begin{equation*}
    \sum_{\sigma,\alpha} \frac{\partial F^j}{\partial
      u_\sigma^\alpha}D_\sigma(\phi^\alpha) = 0,\quad j = 1,\dots,r.
  \end{equation*}
  A horizontal $(n-1)$-form
  \begin{equation*}
    \omega = \sum_i a_i\,dx^1\wedge \dots\wedge\,dx^{i-1}\wedge
    \,dx^{i+1}\wedge \dots\wedge \,dx^n
  \end{equation*}
  defines a conservation law of~$\mathcal{E}$ if
  \begin{equation*}
    \sum_i(-1)^{i+1}D_i(a_i) = 0.
  \end{equation*}
  We are interested in nontrivial conservation laws, i.e., such that~$\omega$
  is not exact.

  Finally, $\mathcal{E}$ is differentially connected if the only solutions of
  the system
  \begin{equation*}
    D_1(f) = \dots = D_n(f) = 0,\quad f \in \mathcal{F}(\mathcal{E}),
  \end{equation*}
  are constants.
\end{coord}

Consider now a locally trivial bundle $\tau\colon \tilde{\mathcal{E}} \to
\mathcal{E}$ such that there exists a flat connection~$\tilde{\mathcal{C}}$ in
$\pi_\infty\circ\tau\colon \tilde{\mathcal{E}} \to
M$. Following~\cite{Trends}, we say that~$\tau$ is a (differential) covering
over~$\mathcal{E}$ if one has
\begin{equation*}
  \tau_*(\tilde{\mathcal{C}}_X) = \mathcal{C}_X
\end{equation*}
for any vector field $X\in D(M)$. Objects existing on~$\tilde{\mathcal{E}}$
are nonlocal for~$\mathcal{E}$: e.g., symmetries of~$\tilde{\mathcal{E}}$ are
nonlocal symmetries of~$\mathcal{E}$, conservation laws
of~$\tilde{\mathcal{E}}$ are nonlocal conservation laws of~$\mathcal{E}$,
etc. A derivation $S\colon \mathcal{F}(\mathcal{E}) \to
\mathcal{F}(\tilde{\mathcal{E}})$ is called a nonlocal shadow if the diagram
\begin{equation*}
  \xymatrix{
    \mathcal{F}(\mathcal{E})\ar[r]^-{\mathcal{C}_X}\ar[d]_-S&
    \mathcal{F}(\mathcal{E})\ar[d]^-S\\ 
    \mathcal{F}(\tilde{\mathcal{E}})\ar[r]^-{\tilde{\mathcal{C}}_X}&
    \mathcal{F}(\tilde{\mathcal{E}}) 
  }
\end{equation*}
is commutative for any~$X\in D(M)$. In particular, any symmetry
of the equation~$\mathcal{E}$, as well as
restrictions~$\eval{\tilde{S}}_{\mathcal{F}(\mathcal{E})}$ of nonlocal
symmetries may be considered as shadows. A nonlocal symmetry is said to be
invisible if its shadow~$\eval{\tilde{S}}_{\mathcal{F}(\mathcal{E})}$ vanishes.

A covering $\tau$ is said to be irreducible if~$\tilde{\mathcal{E}}$ is
differentially connected. Two coverings are equivalent if there exists a
diffeomorphism $g\colon \tilde{\mathcal{E}}_1 \to \tilde{\mathcal{E}}_2$ such
that the diagrams
\begin{equation*}
  \xymatrix{
    \tilde{\mathcal{E}}_1\ar[rr]^g\ar[rd]_-{\tau_1}&&
    \tilde{\mathcal{E}}_2\ar[ld]^-{\tau_2}\\ 
    &\mathcal{E}\rlap{,}&
  }\qquad
  \xymatrix{
    D(\tilde{\mathcal{E}}_1)\ar[rr]^{g_*}&&D(\tilde{\mathcal{E}}_2)\\
    &D(M)\ar[lu]^-{\tilde{\mathcal{C}}_1}\ar[ru]_-{\tilde{\mathcal{C}}_2}&
  }
\end{equation*}
are commutative. Note also that for any two coverings their Whitney product is
naturally defined. A covering is called linear if~$\tau$ is a vector bundle
and the action of vector fields~$\tilde{\mathcal{C}}_X$ preserves the subspace
of fiber-wise linear functions in~$\mathcal{F}(\tilde{\mathcal{E}})$.

In the case of 2D equations, there exists a fundamental relation between
special type of coverings over~$\mathcal{E}$ and conservation laws of the
latter. Let~$\tau$ be a covering of rank~$l<\infty$. We say that~$\tau$ is an
Abelian covering if there exist~$l$ independent conservation
laws~$[\omega_i] \in H_h^1(\mathcal{E})$, $i=1,\dots,l$, such that the
forms~$\tau^*(\omega_i)$ are exact. Then equivalence classes of such coverings
are in one-to-one correspondence with $l$-dimensional $\mathbb{R}$-subspaces
in~$H_h^1(\mathcal{E})$.

\begin{coord}
  Choose a trivialization of the covering~$\tau$ and let $w^1,\dots,w^l,\dots$
  be coordinates in fibers (the are called nonlocal variables). Then the
  covering structure is given by the extended total derivatives
  \begin{equation*}
    \tilde{D}_i = D_i + X_i,\quad i= 1,\dots,n,
  \end{equation*}
  where
  \begin{equation*}
    X_i = \sum_\alpha X_i^\alpha\frac{\partial}{\partial w^\alpha}
  \end{equation*}
  are $\tau$-vertical vector fields (nonlocal tails) enjoying the condition
  \begin{equation}\label{eq:2}
    D_i(X_j) - D_j(X_i) + [X_i,X_j] = 0,\quad i<j.
  \end{equation}
  Here $D_i(X_j)$ denotes the action of $D_i$ on coefficients
  of~$X_j$. Relations~\eqref{eq:2} (flatness of~$\tilde{\mathcal{C}}$) amount
  to the fact that the manifold~$\tilde{\mathcal{E}}$ endowed with the
  distribution~$\tilde{\mathcal{C}}$ coincides with the infinite prolongation
  of the overdetermined system
  \begin{equation*}
    \frac{\partial w^\alpha}{\partial x^i} = X_i^\alpha,
  \end{equation*}
  which is compatible modulo~$\mathcal{E}$.

  Irreducible coverings are those for which the system of vector fields
  $\tilde{D}_1, \dots, \tilde{D}_n$ has no nontrivial
  integrals. If~$\bar{\tau}$ is another covering with the nonlocal
  tails~$\bar{X}_i = \sum\bar{X}_i^\beta\partial/\partial \bar{w}^\beta$, then
  the Whitney product $\tau \oplus \bar{\tau}$ of~$\tau$ and~$\bar{\tau}$ is
  given by
  \begin{equation*}
    \tilde{D}_i = D_i + \sum_\alpha X_i^\alpha\frac{\partial}{\partial
      w^\alpha} + \sum_\beta  \bar{X}_i^\beta\frac{\partial}{\partial
      \bar{w}^\beta}.
  \end{equation*}
  A covering is Abelian if the coefficients~$X_i^\alpha$ are independent of
  nonlocal variables~$w^j$. If $n=2$ and
  $\omega_\alpha = X_1^\alpha\,dx^1 + X_2^\alpha\,dx^2$, $\alpha=1,\dots,l$,
  are conservation laws of~$\mathcal{E}$ then the corresponding Abelian
  covering is given by the system
  \begin{equation*}
    \frac{\partial w^\alpha}{\partial x^i} = X_i^\alpha,\qquad i=1,2,\quad
    \alpha = 1,\dots,l,
  \end{equation*}
  or
  \begin{equation*}
    \tilde{D}_i = D_i + \sum_\alpha X_i^\alpha\frac{\partial}{\partial w^\alpha}.
  \end{equation*}
  Vice versa, is such a covering is given, then one can construct the
  corresponding conservation law.

  The horizontal de~Rham differential on~$\tilde{\mathcal{E}}$ is $\tilde{d}_h
  = \sum_i dx^i\wedge \tilde{D}_i$. A covering is linear if
  \begin{equation}\label{eq:7}
    X_i^\alpha = \sum_\beta X_{i,\beta}^\alpha w^\beta,
  \end{equation}
  where~$X_{i,\beta}^\alpha \in \mathcal{F}(\mathcal{E})$.

  \begin{remark}\label{sec:preliminaries-rem-1}
    Denote by~$\mathbf{X}_i$ the $\mathcal{F}(\mathcal{E})$-valued
    matrix~$(X_{i,\beta}^\alpha)$ that appears in~\eqref{eq:7}. Then
    Eq.~\eqref{eq:2} may be rewritten as
    \begin{equation*}
      D_i(\mathbf{X_j}) - D_j(\mathbf{X_i}) + [\mathbf{X_i,\mathbf{X_j}}] =
      0.
    \end{equation*}
    for linear coverings. Thus, a linear covering defines a zero-curvature
    reperesentation for~$\mathcal{E}$ and vice versa.
  \end{remark}

  A nonlocal symmetry in~$\tau$ is a vector field
  \begin{equation*}
    S_{\phi,\psi} = \sum \tilde{D}_\sigma(\phi^j)\frac{\partial}{\partial
      u_\sigma^j} + \sum \psi^\alpha\frac{\partial}{\partial w^\alpha}, 
  \end{equation*}
  where the vector functions~$\phi = (\phi^1,\dots,\phi^m)$ and $\psi =
  (\psi^1,\dots, \psi^\alpha,\dots)$ on~$\tilde{\mathcal{E}}$ satisfy the
  system of equations
  \begin{align}
    \label{eq:3}
    &\sum\frac{\partial F^j}{\partial u_\sigma^j}\tilde{D}_\sigma(\phi^j) = 0,\\
    \label{eq:4}
    &\tilde{D}_i(\psi^\alpha) = \sum\frac{\partial X_i^\alpha}{\partial
    u_\sigma^j}\tilde{D}_\sigma(\phi^j) + \sum\frac{\partial
    X_i^\alpha}{\partial w^\beta}\psi^\beta.
  \end{align}
  Nonlocal shadows are the derivations
  \begin{equation*}
    \tilde{\Ev}_\phi =  \sum \tilde{D}_\sigma(\phi^j)\frac{\partial}{\partial
      u_\sigma^j},
  \end{equation*}
  where~$\phi$ satisfies Eq.~\eqref{eq:3}, invisible symmetries are
  \begin{equation*}
    S_{\phi,0} = \sum \psi^\alpha\frac{\partial}{\partial w^\alpha},
  \end{equation*}
  where $\psi$ satisfies
  \begin{equation}
    \label{eq:5}
    \tilde{D}_i(\psi^\alpha) = \sum\frac{\partial
      X_i^\alpha}{\partial w^\beta}\psi^\beta.
  \end{equation}
  In what follows, we use the notation $\taui{\tau}\colon
  \taui{\tilde{{\mathcal{E}}}} \to \tilde{\mathcal{E}}$ for the covering
  defined by~Eq.~\eqref{eq:5}.
\end{coord}

\begin{remark}
  Eq.~\eqref{eq:5} defines a linear covering over~$\tilde{\mathcal{E}}$. Due
  to Remark~\ref{sec:preliminaries-rem-1}, we see that for any non-Abelian
  covering we obtain in such a way a nonlocal zero-curvature representation
  with the matrices~$\mathbf{X}_i = (\partial X_i^\alpha/\partial w^\beta)$.
\end{remark}

\begin{remark}
  The covering $\taui{\tau}\colon \taui{\tilde{{\mathcal{E}}}} \to
  \tilde{\mathcal{E}}$ is the vertical part of the tangent covering
  $\mathbf{t}\colon \mathcal{T}\tilde{\mathcal{E}}\to \tilde{\mathcal{E}}$,
  see the definition in~\cite{KVV}.
\end{remark}

\section{The main result}\label{sec:main-result}

From now on we consider two-dimensional scalar equations with the independent
variables~$x$ and~$y$. We shall show that any such an equation that admits an
irreducible covering possesses a (nonlocal) conservation law.

\begin{example}
  \label{sec:main-result-exmpl-1}
  Let us revisit the Bianchi example discussed in the beginning of the
  paper. Equations~\eqref{eq:1} define a one-dimensional non-Abelian
  covering~$\tau\colon \tilde{\mathcal{E}}=\mathcal{E}\times\mathbb{R} \to
  \mathcal{E}$ over the sine-Gordon equation~\eqref{eq:6} with the nonlocal
  variable~$w$. Then the defining equations~\eqref{eq:5} for invisible
  symmetries in this covering are
  \begin{equation*}
    \frac{\partial\psi}{\partial x} = -\cos(u+w)\psi,\quad
    \frac{\partial\psi}{\partial y} = -\cos(u-w)\psi.
  \end{equation*}
  This is a one-dimensional linear covering over~$\tilde{\mathcal{E}}$ which
  is equivalent to the Abelian covering
  \begin{equation*}
    \frac{\partial\bar{\psi}}{\partial x} = -\cos(u+w),\quad
    \frac{\partial\bar{\psi}}{\partial y} = -\cos(u-w),
  \end{equation*}
  where~$\bar{\psi} = \ln\psi$. Thus, we obtain the nonlocal conservation law
  \begin{equation*}
    \omega = -\cos(u+w)\,dx -\cos(u-w)\,dy
  \end{equation*}
  of the sine-Gordon equation.
\end{example}

The next result shows that Bianchi's observation is of a quite general nature.

\begin{proposition}
  \label{sec:main-result-lemma-1}
  Let $\tau\colon \tilde{\mathcal{E}} \to \mathcal{E}$ be a one-dimensional
  non-Abelian covering over~$\mathcal{E}$. Then\textup{,} if $\tau$ is
  irreducible\textup{,} $\taui{\tau}\colon \taui{\tilde{\mathcal{E}}} \to
  \mathcal{E}$ defines a nontrivial conservation law of the
  equation~$\tilde{\mathcal{E}}$ \textup{(}and\textup{,}
  consequently\textup{,} of~$\mathcal{E}$ too\textup{)}.
\end{proposition}

\begin{proof}
  Consider the total derivatives
  \begin{align*}
    D_x^{\mathbf{I}}&= \tilde{D}_x + \frac{\partial X}{\partial w}
                      \psi\frac{\partial}{\partial \psi} = D_x +
                      X\frac{\partial}{\partial w} +
                      \frac{\partial X}{\partial w} 
                      \psi\frac{\partial}{\partial \psi}
    \\
    D_y^{\mathbf{I}}&=\tilde{D}_y + \frac{\partial Y}{\partial w}
                      \psi\frac{\partial}{\partial \psi} = D_y +
                      Y\frac{\partial}{\partial w} +
                      \frac{\partial Y}{\partial w} 
                      \psi\frac{\partial}{\partial \psi}
  \end{align*}
  on $\taui{\mathcal{E}}$ and assume
  that~$a\in \mathcal{F}(\tilde{\mathcal{E}})$ is a common nontrivial integral
  of these fields:
  \begin{equation}\label{eq:9}
     D_x^{\mathbf{I}}(a) =  D_y^{\mathbf{I}}(a) =0, \quad a\neq\const.
  \end{equation}
  Choose a point in $\taui{\mathcal{E}}$ and assume that the formal series
  \begin{equation}\label{eq:10}
    a_0 + a_1\psi + \dots + a_j\psi^j + \dots,\quad a_j\in
    \mathcal{F}(\tilde{\mathcal{E}}), 
  \end{equation}
  converges to $a$ in a neighborhood of this point. Substituting
  relations~\eqref{eq:10} to~\eqref{eq:9} and equating coefficients at the
  same powers of~$\psi$, we get
  \begin{equation*}
    \tilde{D}_x(a_j) +j\frac{\partial X}{\partial w}a_j =0,\quad
    \tilde{D}_y(a_j) +j\frac{\partial Y}{\partial w}a_j =0,\qquad
    j=0,1,\dots,
  \end{equation*}
  and, since $\tau$ is irreducible, this implies that $a_0 = k_0 = \const$ and
  \begin{equation*}
    \frac{\tilde{D}_x(a_j)}{a_j} = j\frac{\tilde{D}_x(a_1)}{a_1},\quad
    \frac{\tilde{D}_y(a_j)}{a_j} = j\frac{\tilde{D}_y(a_1)}{a_1}.
  \end{equation*}
  Hence, $a_j = k_j(a_1)^j$, $j>0$. Substituting these relations
  to~\eqref{eq:10}, we see that~$a=a(\theta)$, where~$\theta=a_1\psi$, $a_1
  \in \mathcal{F}(\mathcal{E})$. Then Equation~\eqref{eq:9} take the form
  \begin{equation*}
    \dot{a}\psi\left(\tilde{D}_x(a_1) + \frac{\partial X}{\partial
        \psi}\right) 
    =0,\quad \dot{a}\psi\left(\tilde{D}_y(a_1) + \frac{\partial
        Y}{\partial\psi}\right) 
    =0, \qquad \dot{a} = \frac{da}{d\theta}.
  \end{equation*}
  Thus
  \begin{equation*}
    \frac{\partial X}{\partial w} = -\tilde{D}_x(a_1),\quad
    \frac{\partial Y}{\partial w} = -\tilde{D}_y(a_1)
  \end{equation*}
  and the function~$w + a_1$ is a nontrivial integral of~$\tilde{D}_x$
  and~$\tilde{D}_y$. Contradiction.

  Finally, repeating the scheme of Example~\ref{sec:main-result-exmpl-1}, we
  pass to the equivalent covering by setting~$\bar{\psi} = \ln\psi$ and obtain
  the nontrivial conservation law
  \begin{equation*}
    \omega = \frac{\partial X}{\partial w}\,dx + \frac{\partial Y}{\partial
      w}\, dy
  \end{equation*}
  on~$\taui{\mathcal{E}}$.
\end{proof}

Indeed, Bianchi's result has a further generalization. To formulate the
latter, let us say that a covering $\tau\colon \tilde{\mathcal{E}} \to
\mathcal{E}$ is strongly non-Abelian if for any nontrivial conservation
law~$\omega$ of the equation~$\mathcal{E}$ its lift~$\tau^*(\omega)$ to the
manifold~$\tilde{\mathcal{E}}$ is nontrivial as well. Now, a straightforward
generalization of Proposition~\ref{sec:main-result-lemma-1} is

\begin{proposition}
  \label{sec:main-result-lemma-2}
  Let~$\tau\colon \tilde{\mathcal{E}} \to \mathcal{E}$ be an irreducible
  covering over a differentially connected equation. Then~$\tau$ is a strongly
  non-Abelian covering if and only if the covering~$\taui{\tau}$ is
  irreducible.
\end{proposition}

We shall now need the following construction. Let $\tau\colon
\tilde{\mathcal{E}} \to \mathcal{E}$ be a linear covering. Consider the
fiber-wise projectivization $\taup{\tau}\colon \taup{\tilde{\mathcal{E}}} \to
\mathcal{E}$ of the vector bundle~$\tau$. Denote by $\mathbf{p}\colon
\tilde{\mathcal{E}} \to\taup{\mathcal{E}}$ the natural projection. Then,
obviously, the projection $\mathbf{p}_*(\tilde{\mathcal{C}})$ is well defined
and is an $n$-dimensional integrable distribution
on~$\taup{\mathcal{E}}$. Thus, we obtain the following commutative diagram of
coverings
\begin{equation*}
  \xymatrix{
    \tilde{\mathcal{E}}\ar[rr]^-{\mathbf{p}}\ar[rd]_-{\tau}&
    &\taup{\mathcal{E}}\ar[ld]^-{\taup{\tau}}\\ 
    &\mathcal{E}\rlap{,}&
  }
\end{equation*}
where $\rank(\mathbf{p}) = 1$ and $\rank(\taup{\tau})= \rank(\tau) - 1$.

\begin{proposition}
  \label{sec:main-result-lemma-3}
  Let $\tau\colon \tilde{\mathcal{E}} \to \mathcal{E}$ be an irredicible
  covering. Then the covering~$\taup{\tau}$ is irreducible as well.
\end{proposition}

\begin{coord}
  Let $\rank(\tau) = l>1$ and
  \begin{equation}\label{eq:8}
    w_{x^i}^\alpha = \sum_{\beta=1}^l X_{i,\beta}^\alpha w^\beta,\qquad i =
    1,\dots,n,\quad \alpha = 1,\dots,l,
  \end{equation}
  be the defining equations of the covering~$\tau$, see
  Eq.~\eqref{eq:7}. Choose an affine chart in the fibers of~$\taup{\tau}$. To
  this end, assume for example that~$w^l\neq0$ and set
  \begin{equation*}
    \bar{w}^\alpha = \frac{w^\alpha}{w^l},\qquad l=1,\dots,l-1,
  \end{equation*}
  in the domain under consideration. Then from Equations~\eqref{eq:8} it
  follows that the system
  \begin{equation*}
    \bar{w}_{x^i}^\alpha = X_{i,l}^\alpha - X_{i,l}^l\bar{w}^\alpha +
    \sum_{\beta=1}^{l-1} X_{i,\beta}^\alpha\bar{w}^\beta - \bar{w}^\alpha
    \sum_{\beta = 1}^{l-1}X_{i,\beta}^l\bar{w}^\beta,\qquad i=1,\dots,n,\quad
    \alpha =1,\dots,l-1.
  \end{equation*}
  locally provides the defining equation for the covering~$\taup{\tau}$.
\end{coord}

We are now ready to state and prove the main result.

\begin{theorem}
  \label{sec:main-result-thm-1}
  Assume that a differentially connected equation~$\mathcal{E}$ admits a
  nontrivial covering $\tau\colon \tilde{\mathcal{E}} \to \mathcal{E}$. Then
  it possesses at least one nontrivial \textup{(}nonlocal\textup{)}
  conservation law.
\end{theorem}

\begin{proof}
  Actually, the proof is a description of a procedure that allows one to
  construct the desired conservation law.

  Note first that we may assume the covering~$\tau$ to be irreducible. Indeed,
  otherwise the space~$\tilde{\mathcal{E}}$ is foliated by maximal integral
  manifolds of the distribution~$\tilde{\mathcal{C}}$. Let~$l_0$ denote the
  codimension of the generic leaf and~$l=\rank(\tau)$. Then
  \begin{itemize}
  \item $l>l_0$, because~$\tau$ is a nontrivial covering;
  \item the integral leaves project to~$\mathcal{E}$ surjectively,
    because~$\mathcal{E}$ is a differentially connected equation.
  \end{itemize}
  This means that in vicinity of a generic point we can consider~$\tau$ as an
  $l_0$-parametric family of irreducible coverings whose rank is~$r =
  l-l_0>0$. Let us choose one of them and denote it by~$\tau_0\colon
  \mathcal{E}_0 \to \mathcal{E}$.

  If~$\tau_0$ is not strongly non-Abelian, then this would mean
  that~$\mathcal{E}$ possesses at least one nontrivial conservation law and we
  have nothing to prove further. Assume now that the covering~$\tau_0$ is
  strongly non-Abelian. Then due to Proposition~\ref{sec:main-result-lemma-2}
  the linear covering~$\taui{\tau}_0$ is irreducible and by
  Proposition~\ref{sec:main-result-lemma-3} its projectivization~$\tau_1 =
  \taup{(\taui{\tau}_0)}$ possesses the same property and~$\rank(\tau_1) =
  r-1$. Repeating the construction, we arrive to the diagram
  \begin{equation*}\xymatrixcolsep{4.5pc}
    \xymatrix{
      &&\ar[d]^-{\mathbf{p}}
      \ar[dl]_-{\taui{\tau}_0}\taui{\mathcal{E}}_0&
      \dots&
      \ar[d]^-{\mathbf{p}}
      \ar[dl]_-{\taui{\tau}_{r-2}}\taui{\mathcal{E}}_{r-2}\\    
      \mathcal{E}&\ar[l]_-{\tau_0}\mathcal{E}_0&
      \ar[l]^-{\tau_1=\taup{(\taui{\tau}_0)}}
      \taup{(\taui{\mathcal{E}}_0)}=\mathcal{E}_1&  
      \ar[l]\dots&\ar[l]^-{\tau_{r-1}=\taup{(\taui{\tau}_{r-2})}}
      \taup{(\taui{\mathcal{E}}_{r-2})} 
      =\mathcal{E}_{r-1}\rlap{,} 
    }
  \end{equation*}
  where~$\rank(\tau_i) = l-i$. Thus, in ~$r-1$ steps at most we shall arrive
  to a one-dimensional irreducible covering and find ourselves in the
  situation of Proposition~\ref{sec:main-result-exmpl-1} and this finishes the
  proof.
\end{proof}

\section{Examples}
\label{sec:examples}
Let us discuss several illustrative examples.

\begin{example}
  \label{sec:main-result-exmpl-2}
  Consider the Korteweg-de~Vries equation in the form
  \begin{equation}\label{eq:11}
    u_t = uu_x + u_{xxx}
  \end{equation}
  and the well known Miura transformation~\cite{Miura}
  \begin{equation*}
    u = w_x - \frac{1}{6}w^2.
  \end{equation*}
  The last formula is a part of the defining equations for the non-Abelian
  covering
  \begin{equation*}
    \begin{array}{rcl}
      w_x&=&u + \dfrac{1}{6}w^2,\\
      w_t&=&u_{xx} + \dfrac{1}{3}wu_x + \dfrac{1}{3}u^2 + \dfrac{1}{18}w^2u,
    \end{array}
  \end{equation*}
  the covering equation being
  \begin{equation*}
    w_t = w_{xxx} - \frac{1}{6}w^2w_x,
  \end{equation*}
  i.e., the modified KdV equation. Then the corresponding
  covering~$\taui{\tau}$ is defined by the system
  \begin{equation*}
    \begin{array}{rcl}
      \psi_x&=&\dfrac{1}{3}w\psi,\\[10pt]
      \psi_t&=&\dfrac{1}{3}\left(u_x + \dfrac{1}{3}wu\right)\psi
    \end{array}
  \end{equation*}
  that, after relabeling~$\psi\mapsto 3\ln \psi$ gives us the nonlocal
  conservation law
  \begin{equation*}
    \omega = w\,dx + \left(u_x + \dfrac{1}{3}wu\right)\,dt
  \end{equation*}
  of the KdV equation.
\end{example}

\begin{example}
  \label{sec:main-result-exmpl-3}
  The well known Lax pair, see~\cite{Lax}, for the KdV equation may be
  rewritten in terms of zero-curvature representation
  \begin{equation*}
    D_x(\mathbf{T}) - D_t(\mathbf{X}) + [\mathbf{X},\mathbf{T}] =0,
  \end{equation*}
  where $(2\times2)$ matrices $\mathbf{X}$ and~$\mathbf{T}$ are of the form
  \begin{equation*}
    \mathbf{X} =
    \begin{pmatrix}
      0&\frac{1}{6}\\[3pt]
      -(\lambda+u)&0
    \end{pmatrix},\qquad
    \mathbf{T} =
    \begin{pmatrix}
      -\frac{1}{6}u_x&\frac{1}{9}\left(\frac{1}{2}u-\lambda\right)\\[3pt]
      -u_{xx}-\frac{1}{3}u^2-\frac{1}{3}\lambda u-\frac{2}{3}\lambda^2&
      \frac{1}{6}u_x
    \end{pmatrix},
  \end{equation*}
  $\lambda\in\mathbb{R}$ being a real parameter. As it follows from
  Remark~\ref{sec:preliminaries-rem-1}, this amounts to existence of the
  two-dimensional linear covering~$\tau$ given by the system
  \begin{align*}
    w_x^1&=\frac{1}{6}w^2,\\
    w_t^1&=-\frac{1}{6}u_xw^1 + \frac{1}{9}\left(\frac{1}{2}-
      \lambda\right)w^2,\\ 
    w_x^2&=-(\lambda+u)w^1,\\
    w_t^2&=-\left(u_{xx} + \frac{1}{3}u^2 + \frac{1}{3}\lambda u +
      \frac{2}{3}\lambda^2\right)w^1 + \frac{1}{6}u_xw^2. 
  \end{align*}
  Let us choose for the affine chart the domain~$w^2\neq0$ and
  set~$\psi=w^1/w^2$. Then the covering~$\taup{\tau}$ is described by the
  system
  \begin{align*}
    \psi_x&=(\lambda+u)\psi + \frac{1}{6},\\
    \psi_t&=\left(u_{xx} + \frac{1}{3}u^2 + \frac{1}{3}\lambda u
      +\frac{2}{3}\lambda^2\right)\psi^2 - \frac{1}{3}u_x\psi +
    \frac{1}{9}\left(\frac{1}{2} - \lambda\right)
  \end{align*}
  while~$\tau_1 = \taui{(\taup{\tau})}$ is given by
  \begin{align*}
    \tilde{\psi}_x&=(\lambda+u)\tilde{\psi},\\
    \tilde{\psi}_t&=2\left(u_{xx} + \frac{1}{3}u^2 + \frac{1}{3}\lambda u
      +\frac{2}{3}\lambda^2\right)\psi\tilde{\psi} -
    \frac{1}{3}u_x\tilde{\psi}.
  \end{align*}
  Thus, we obtain the conservation law
  \begin{equation*}
    \omega = (\lambda+u)\,dx + \left(2\left(u_{xx} + \frac{1}{3}u^2 +
        \frac{1}{3}\lambda u 
        +\frac{2}{3}\lambda^2\right)\psi -
      \frac{1}{3}u_x\right)\,dt
  \end{equation*}
  that depends on the nonlocal variable~$\psi$.
\end{example}

\begin{example}
  \label{sec:examples-1}
  Consider the potential KdV equation in the form
  \begin{equation*}
    u_t = 3u_x^2 + u_{xxx}
  \end{equation*}
  Its B\"{a}cklund auto-transformation is associated to the covering~$\tau$
  \begin{align*}
    w_x &= \lambda - u_x -\frac{1}{2}(w-u)^2,\\
    w_t &= 2\lambda^2 - 2\lambda u_x - u_x^2 - u_{xxx} + 2u_{xx}(w-u) -
    (\lambda + u_x)(w-u)^2,
  \end{align*}
  where~$\lambda\in\mathbb{R}$, see~\cite{W-E}. Then the
  covering~$\taui{\tau}$ is
  \begin{align*}
    \psi_x &= -(w-u)\psi,\\
    \psi_t &= 2\big(u_{xx}\psi - (\lambda + u_x)(w-u)\big)\psi,
  \end{align*}
  which leads to the nonlocal conservation law
  \begin{equation*}
    \omega  = -(w-u)\,dx + 2\big(u_{xx}\psi - (\lambda + u_x)(w-u)\big)\,dt
  \end{equation*}
  of the potential KdV equation.
\end{example}

\begin{example}
  \label{sec:examples-2}
  The Gauss-Mainardi-Codazzi equations read
  \begin{equation}\label{eq:12}
    u_{xy}  = \frac{g - fh}{\sin u},\qquad
    f_y = g_x + \frac{h - g\cos u}{\sin u}u_x,\qquad
    g_y = h_x - \frac{f - g\cos u}{\sin u}u_y,
  \end{equation}
  see~\cite{Sym}. This is an under-determined system, and imposing additional
  conditions on the unknown functions~$u$, $f$, $g$, and $h$ one obtains
  equations that describe various types of surfaces in~$\mathbb{R}^2$,
  cf.~\cite{K-M-Gauss}. System~\eqref{eq:12} always admits the following
  $\mathbb{C}$-valued zero-curvature representation
  \begin{equation*}
    D_x(\mathbf{Y}) - D_y(\mathbf{X}) +[\mathbf{X},\mathbf{Y}] = 0
  \end{equation*}
  with the matrices
  \begin{equation*}
    \mathbf{X} = \frac{\mathrm{i}}{2}
    \begin{pmatrix}
      u_x & \dfrac{\mathrm{e}^{\mathrm{i}u}f -g}{\sin u}\\
      \dfrac{\mathrm{e}^{-\mathrm{i}u}f -g}{\sin u} & -u_x
    \end{pmatrix},\qquad
    \mathbf{Y} = \frac{\mathrm{i}}{2}
    \begin{pmatrix}
      0 & \dfrac{\mathrm{e}^{iu}g - h}{\sin u} \\
      \dfrac{\mathrm{e}^{-\mathrm{i}u}g - h}{\sin u} & 0
    \end{pmatrix}
  \end{equation*}
  The corresponding two-dimensional linear covering~$\tau$ is defined by the
  system
  \begin{equation*}
    \begin{array}{rcl}
      w_x^1 &=& u_xw^1 + \dfrac{\mathrm{e}^{\mathrm{i}u}f -g}{\sin u}w^2,\\[10pt]
      w_y^1 &=& \dfrac{\mathrm{e}^{iu}g - h}{\sin u}w^2,
    \end{array}\qquad
    \begin{array}{rcl}
      w_x^2&=&\dfrac{\mathrm{e}^{-\mathrm{i}u}f -g}{\sin u}w^1 -u_xw^2,\\[10pt]
      w_y^2&=&\dfrac{\mathrm{e}^{-\mathrm{i}u}g - h}{\sin u}w^1.
  \end{array}
  \end{equation*}
  Hence, the covering~$\taup{\tau}$ in the domain $w^2\neq 0$ is
  \begin{equation*}
    \psi_x = \frac{\mathrm{e}^{\mathrm{i}u}f - g}{\sin u} + 2u_x\psi -
    \frac{\mathrm{e}^{-\mathrm{i}u}f - g}{\sin u}\psi^2,\qquad
    \psi_y = \frac{\mathrm{e}^{\mathrm{i}u}g - h}{\sin u} -
    \frac{\mathrm{e}^{-\mathrm{i}u}g - h}{\sin u}\psi^2.
  \end{equation*}
  Thus, the covering~$\taui{(\taup{\tau})}$, given by
  \begin{equation*}
    \tilde{\psi}_x = 2\left(u_x -
      \frac{\mathrm{e}^{-\mathrm{i}u}f - g}{\sin u}\psi\right)\tilde{\psi},\qquad
    \tilde{\psi}_y =  - 2\frac{\mathrm{e}^{-\mathrm{i}u}g - h}{\sin
      u}\psi\tilde{\psi}, 
  \end{equation*}
  defines the nonlocal conservation law
  \begin{equation*}
    \omega = \left(u_x -
      \frac{\mathrm{e}^{-\mathrm{i}u}f - g}{\sin u}\psi\right)\,dx
    - \frac{\mathrm{e}^{-\mathrm{i}u}g - h}{\sin
      u}\psi\,dy
  \end{equation*}
  of the Gauss-Mainardi-Codazzi equations.
\end{example}

\begin{example}
  \label{sec:examples-3}
  The last example shows that the above described techniques fail for
  infinite-dimensional coverings (such coverings are typical for equations of
  dimension greater than two).

  Consider the equation
  \begin{equation*}
    u_{yy} = u_{tx} + u_y u_{xx} - u_x u_{xy} 
  \end{equation*}
  that arises in the theory of integrable hydrodynamical chains,
  see~\cite{Pavlov}. This equation admits the covering~$\tau$ with the
  nonlocal variables~$w^i$, $i=0,1,\dots$, that enjoy the defining relations
  \begin{align*}
    &w_t^0 + u_y w_x^1 = 0,\quad w_y^0 + u_x w_x^1 = 0,\\
    &w_x^i = w^{i+1},\qquad i\geq 0,\\
    &w_t^i + D_x^i(u_y w_x^1) = 0,\quad w_y^i + D_x^i(u_x w_x^1) = 0,\qquad
    i\geq1. 
  \end{align*}
  see~\cite{Comparative}. This is a linear covering, but its projectivization
  does not lead to construction of conservation laws.
\end{example}

\section*{Discussion}
\label{Discussion}

We described a procedure that allows one to associate, in an algorithmic way,
with any nontrivial finite-dimensional covering over a differentially
connected equation a nonlocal conservation law. Nevertheless, this method
fails in the case of infinite-dimensional coverings. It is unclear, at the
moment at least, whether this is an immanent property of such coverings or a
disadvantage of the method. I hope to clarify this in future research.

\section*{Acknowledgments}
\label{sec:acknowledgements}

I am grateful to Michal Marvan, who attracted my attention to the paper by
Luigi Bianchi~\cite{Bianchi}, and to Raffaele Vitolo, who helped me with
Italian. I am also grateful to Valentin Lychagin for a fruitful discussion.

\end{document}